\documentclass[conference]{IEEEtran}
\usepackage{algorithmicx,algorithm}
\usepackage{epsfig,cite,stfloats,bm,amssymb,amsmath,color,subfigure,graphics,extpfeil,amsthm,amsfonts,dsfont,enumerate}
\usepackage{epstopdf}
\usepackage{array}
\usepackage[noend]{algpseudocode}
\newcolumntype{L}[1]{>{\raggedright\let\newline\\\arraybackslash\hspace{0pt}}m{#1}}
\newcolumntype{C}[1]{>{\centering\let\newline\\\arraybackslash\hspace{0pt}}m{#1}}
\newcolumntype{R}[1]{>{\raggedleft\let\newline\\\arraybackslash\hspace{0pt}}m{#1}}
\long\def\symbolfootnote[#1]#2{\begingroup
\def\thefootnote{\fnsymbol{footnote}}
\footnote[#1]{#2}\endgroup}
\psfull

\newtheorem{pro}{Proposition}

\begin{document}

\title{Learn and Pick Right Nodes to Offload}

\author{\IEEEauthorblockN{Zhaowei Zhu, Ting Liu, Shengda Jin,
and Xiliang Luo}
\IEEEauthorblockA{
School of Information Science and Technology, ShanghaiTech University,
Shanghai, China \\
Email: {\{zhuzhw, liuting, jinshd, luoxl\}}@shanghaitech.edu.cn}}

\maketitle

\begin{abstract}
Task offloading is a promising technology to exploit the benefits of fog computing.
An effective task offloading strategy is needed to utilize the computational resources
efficiently. In this paper, we endeavor to seek an online task offloading strategy
to minimize the long-term latency. In particular,
we formulate a stochastic programming problem, where the expectations of the
system parameters change abruptly at unknown time instants. Meanwhile, we consider
the fact that the queried nodes can only feed back the processing results after
finishing the tasks. We then put forward an effective algorithm to solve this
challenging stochastic programming under the non-stationary bandit model.
We further prove that our proposed algorithm is asymptotically optimal in a
non-stationary fog-enabled network. Numerical simulations are carried out to
corroborate our designs.
\end{abstract}


\vspace{-0.1cm}
\section{Introduction}
\vspace{-0.1cm}

With the ever-increasing demands for intelligent services, devices
such as the smart phones are facing challenges in both battery
life and computing power \cite{2013_Dinh_survey}. Rather than offloading
computation to remote clouds, fog computing distributes computing, storage,
control, and communication services along the Cloud-to-Thing continuum
\cite{2009_Satyanarayanan_cloud,2015_Chiang_Fog_IoT}.

In recent years, task offloading becomes a promising technology and attracts
significant attentions from researchers. In general, tasks with high-complexity are
usually offloaded to other nodes such that the battery lifetime and computational
resources of an individual user can be saved \cite{2013_Barbera_bandwidth_energy_MCC}.
For example, ThinkAir \cite{2012_Kosta_ThinkAir} provided a code offloading framework,
which was capable of on-demand computational resource allocation and parallel
execution for each task. In some literatures, the task offloading was modeled as
a deterministic optimization problem, e.g. the maximization of energy efficiency
in \cite{2017_Yang_MEETS}, the joint minimization of energy and latency in
\cite{2017_Dinh_MobileEdge}, and the minimization of energy consumption under
delay constraints in \cite{2017_You_MobileEdge}.
However, one task offloading strategy needs to rely on the real-time states of
the users and the servers, e.g. the length of the computation queue.
From this aspect, the task offloading is a typical stochastic programming problem
and the conventional optimization methods with deterministic parameters are not
applicable. To circumvent this dilemma, the Lyapunov optimization method was
invoked in
\cite{2015_Kwak_JSAC_DREAM,2017_Mao_MobileEdge,2017_Pu_D2D_fog,2017_Yang_DEBTS}
to transform the challenging stochastic programming problem to a sequential
decision problem, which included a series of deterministic problems in each time
slot. Besides, the authors in \cite{2015_Chen_Mobile_Cloud_Computing} provided one
game-theoretic decentralized approach, where each user can make offloading decisions
autonomously.

The aforementioned task offloading schemes all assumed the availability of
perfect knowledge about the system parameters. However, there are some
cases where these parameters are unknown or partially known at the user.
For example, some particular values (a.k.a. bandit feedbacks) are only revealed
for the nodes that are queried. Specifically, the authors in
\cite{2017_Chen_GG_BanditCVX} treated the communication delay and the computation
delay of each task as a posteriori. In \cite{2014_Tekin_ExpertLearning}, the
mobility of each user was assumed to be unpredictable.
When the number of nodes that can be queried is limited due to the finite available
resources, there exists a tradeoff between exploiting the empirically best node
as often as possible and exploring other nodes to find more profitable actions
\cite{2002_Auer_UCB,2011_non-stationary_UCB}.
To balance this tradeoff, one popular approach is to model the exploration versus
exploitation dilemma as a multi-armed bandit (MAB) problem, which has been
extensively studied in statistics \cite{1985_Berry_bandit}.

There are very few prior works addressing this exploration vs. exploitation
tradeoff during task offloading in a fog-enabled network. In this paper, we
assume the processing delay of each task is unknown when we start to process
the task and endeavor to find an efficient task offloading scheme
with bandit feedback to minimize the user's long-term latency.
Our main contributions are as follows. Firstly, we introduce a non-stationary
bandit model to capture the unknown latency variation, which is more practical
than the previous model-based ones, e.g.
\cite{2017_Dinh_MobileEdge,2017_You_MobileEdge,2015_Kwak_JSAC_DREAM,
2017_Mao_MobileEdge,2017_Yang_DEBTS}. Secondly, an efficient task offloading
algorithm is put forward based on the upper-confidence bound (UCB) policy.
Note our proposed scheme is not a straightforward application of the UCB
policy and thus the conventional analysis is not applicable.
We also provide performance guarantees for the proposed algorithm.

The rest of this paper is organized as follows.
Section \ref{Sec:System_model} introduces the task offloading model and system assumptions.
Section \ref{Algorithm} presents one efficient algorithm and the corresponding performance guarantee.
Numerical results are presented in Section \ref{Simulation} and Section \ref{Conclusion} concludes the paper.

\noindent{\it Notations}:
Notation $|\mathcal A|$, ${\rm Unif}(a,b)$, $\mathbb E(A)$, and $\mathbb P(A)$ stand for the cardinality of set $\mathcal A$, the uniform distribution on $(a,b)$, the expectation of random variable $A$, and the probability of event $A$.
Notation $A_n\overset{\rm a.s.}{\longrightarrow}A$ indicates the sequence $\{A_n\}_{n=1}^{\infty}$ converges almost surely towards $A$.
One indicator function $\mathds{1}{\{ \cdot \}}$ takes the value of $1(0)$ when the specified condition is met (otherwise).

\section{System Model} \label{Sec:System_model}

\subsection{Network Model}

\begin{figure}[t]
\centering
\includegraphics[width = 0.45\textwidth]{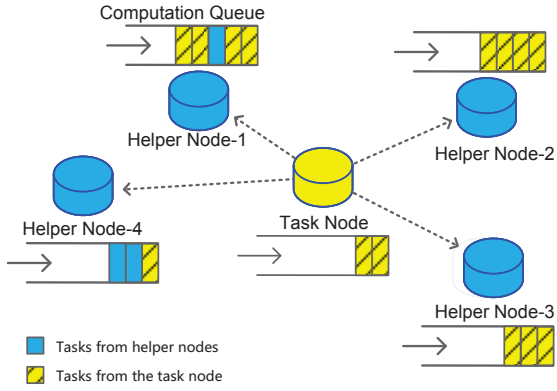}
\caption{A fog-enabled network.
Different colors indicate the tasks from different types of nodes.
The task node is busy dealing with computation tasks, some of which are offloaded to its nearby helper nodes, i.e. helper node-$1$, $2$, $3$, and $4$.
}
\label{Fig:Fog_topology}
\end{figure}

We are interested in a fog-enabled network (see also Fig.~\ref{Fig:Fog_topology})
where both task nodes and helper nodes co-exist. Computation tasks are generated
at each fog node. Each fog node can also communicate with nearby nodes.
The unfinished tasks are assumed to be cached in a first-input first-output (FIFO)
queue at each node. Due to the limited computation and storage resources within one
individual node, the tasks that are processed locally usually experience high
latency, which degrades the quality of service (QoS) and the quality of experience
(QoE). To enable low-latency processing, one task node may offload some of its
computation tasks to the nearby helper nodes. These helper nodes typically possess
more computation and storage resources and are deployed to help other task nodes
on demand. In typical applications such as the online gaming, the tasks are usually
generated periodically and cannot be split arbitrarily. Thus we assume one task is
generated at the beginning of each time slot. Meanwhile, it can be allocated as one
whole piece to one neighboring helper node.

Our goal is to minimize the long-term latency at a particular task node. In particular,
the set of $K$ fog nodes can be classified as
\vspace{-0.2cm}
\begin{equation}
\mathcal I := \left\{ \underbrace{1,2,\cdots,K-1}_{\rm Helper~nodes},
\underbrace{K}_{\rm Task~node} \right\}.	
\end{equation}
In this paper, we assume the task node cannot offload tasks to a helper node
when it is communicating with others. We also assume each task is generated
independently and the task nodes do not cooperate with each other\footnote{
The cooperation among multiple task nodes is beyond the scope of the current
paper and is left for our future works.}.

We use $T(i)$ to represent the amount of time needed to deliver one
bit of information to node-$i$. It is a distance-dependent value
and can be measured before transmission\footnote{We assume different
task nodes occupy pre-allocated orthogonal time or spectrum resources for
the communication to the helper nodes, e.g. TDMA or FDMA.
Note the optimal time/spectrum reusing is itself a non-trivial research
problem \cite{2018_Zhu_D2D}.}.
Denote the data length of task-$t$ by $L_t$.
We also assume the task size is such that the transmission delay $L_t T(i)$
is no more than one time slot. Note the transmission delay is zero
for a locally processed task, i.e. $L_t T(K) = 0$.

Let $Q_t(i)$ denote the queue length of node-$i$ at the beginning of time slot-$t$.
Meanwhile, we denote the time needed to process one bit waiting in the queue at
node-$i$ by $W_t(i)$, and denote the time needed to process one bit in task-$t$
at node-$i$ by $P_t(i)$ when all the tasks ahead in the queue are completed.
Furthermore, we treat $W_t(i)$ and $P_t(i)$ as random variables in this paper.
Accordingly, the expectations are defined as:
\begin{equation}\label{Eq:Expect_W_P}
	\mu_t^W(i):=\mathbb E[W_t(i)], \ \  \mu_t^P(i):=\mathbb E[P_t(i)].
\end{equation}

We assume the total latency of each task is dominated by the delays mentioned above,
i.e. the transmission delay $L_t T(i)$, the waiting delay $Q_t(i) W_{t}(i)$ in the
queue, and the processing delay $L_t P_{t}(i)$. We ignore the latency introduced
during the transmitting of the computing results. Therefore, the total latency
when allocating task-$t$ to node-$i$ can be written as follows.
\begin{equation}\label{Eq:Latency}
	U_t(i) := L_t  T(i) + Q_t(i) W_{t}(i) + L_t P_{t}(i).
\end{equation}

Before we proceed further, here we make the following assumptions:
\begin{itemize}
\item AS-1:
The total latency $U_t(i)$ is unknown before the task is completed;

\item AS-2: The queue length $Q_t(i)$ is broadcasted by node-$i$ at the
beginning of each slot and is available for all the nearby fog nodes;

\item AS-3: The waiting delay and the processing delay, i.e. $W_{t}(i)$
and $P_{t}(i)$, follow unknown distributions. The corresponding expectations,
i.e. $\mu_t^W(i)$ and $\mu_t^P(i)$, change abruptly at unknown time instants
(a.k.a. breakpoints).
\end{itemize}

Different from the model-based task offloading problems addressed in \cite{2017_Dinh_MobileEdge,2017_You_MobileEdge,2015_Kwak_JSAC_DREAM,
2017_Mao_MobileEdge,2017_Yang_DEBTS}, we do not require any specific
relationships between the CPU frequencies and the processing delays in
AS-1 and AS-3 as in \cite{2017_Chen_GG_BanditCVX}. This is a more realistic
setting due to the following reasons.
Firstly, the data lengths and the computation complexities of tasks should
be modeled as a sequence of independent random variables.
This is because their distributions may change abruptly and be completely
different due to the changes in task types.
Additionally, the computation capability, e.g. CPU frequencies, CPU cores,
and memory size, of each node is different and may also follow abruptly-changing
distributions. All of these uncertainties mentioned above make it very tough for
an individual node to forecast the amount of time spent in processing different tasks.
It also costs a lot of overheads for an individual node to obtain the global
information about the whole system.
As a result, the processing delay and the waiting delay cannot be calculated
accurately in a practical system with the conventional model, where the delays
are simply determined by the data length and the configured CPU frequency
\cite{2017_Mao_MobileEdge}.

In our paper, the processing delay and the waiting delay are only reported
after the corresponding task is finished. Namely, the observations of the waiting
delay $\tau_t^{W}$ and the processing delay $\tau_t^{P}$ are treated as posterior
information. Note these delays can be obtained via the timestamp feedback from
the corresponding node after finishing task-$t$. Accordingly, we obtain the
realizations of $W_t(i)$ and $P_t(i)$ as
\begin{equation}
	w_t(i) = \frac{\tau_t^{W}}{Q_t(i)}\mathds 1{\{I_t=i\}}, \ p_t(i) = \frac{\tau_t^{P} }{L_t}\mathds 1{\{I_t=i\}}.
\end{equation}

\subsection{Problem Formulation}

The general minimization of the long-term average latency of tasks can
be formulated as follows.
\begin{equation}\label{Eq:GlobalOpt}
\begin{split}
\underset{\{I_t,\forall t\}}{\rm minimize} \  &
	\lim_{T\rightarrow \infty} \frac{1}{T}\sum_{t=1}^T  \sum_{i=1}^K U_t(i)\mathds 1{\{I_t=i\}}  \\
	{\rm subject \ to}  & \ \ I_t \in\mathcal I,t=1,2,\cdots,T,
\end{split}
\end{equation}
where $I_t$ represents the index of the node to process task-$t$. There are two
difficulties in solving the above problem. Firstly, it is a stochastic programming
problem. The exact information about the latency $U_t(i)$ is not available before
the $t$-th task is completed. Additionally, even if $U_t(i)$ is known apriori,
this problem is still a combinatorial optimization problem and the complexity is
in the order of $\mathcal O(K^T)$. This is due to the fact that the previous
offloading decisions determine the queue length in each fog node and further affect
the decisions of future tasks. See \cite{2014_Tekin_ExpertLearning} for an example.
To render the task offloading strategies welcome online updating, one popular way
is to convert this challenging stochastic and combinatorial optimization problem
into one low-complexity sequential decision problem at each time slot
\cite{2015_Kwak_JSAC_DREAM,2017_Mao_MobileEdge,2017_Pu_D2D_fog,2017_Yang_DEBTS}.
Given the task offloading decisions made in the previous $(t-1)$ time slots, the
optimal strategy turns into allocating task-$t$ to the node with minimal latency
at time slot-$t$. Meanwhile, under the stochastic framework \cite{2012_Bubeck_Bandit},
it is more natural to focus on the expectation, i.e. $\mathbb E[U_t(i)]$.
Accordingly, the problem in (\ref{Eq:GlobalOpt}) becomes the following
one in the $t$-th time slot:
\begin{equation}\label{Eq:OptEach}
\begin{split}
\underset{I_t \in\mathcal I}{\rm minimize} \  &
	\sum_{i\in\mathcal I} \mathbb E[U_t(i)]\mathds 1{\{I_t=i\}}
\end{split}
\end{equation}

However, the above formulation is still a stochastic programming problem.
Although the tasks offloaded previously do enable an empirical average as an
estimate of the expectation $\mathbb E[U_t(i)]$, this information may be
inaccurate due to limited number of observations. Note the information about
node-$i$ is from the feedbacks from node-$i$ when it finishes the corresponding tasks.
In order to get more information about one specific node, the task node has to
offload more tasks to that node even though it may not be the empirically best
node to offload. Therefore, an exploration-exploitation tradeoff exists in this
problem. In the following parts, we endeavor to find one efficient scheme to
solve the problem in (\ref{Eq:OptEach}).

\section{Efficient Offloading Algorithm} \label{Algorithm}

\subsection{Task Offloading with Discounted-UCB}

To strike a balance between the aforementioned exploration and exploitation,
we model the task offloading as a non-stationary multi-armed bandit (MAB) problem
\cite{2011_non-stationary_UCB}, where each node in $\mathcal I$ is regarded as one
arm. When a particular task is generated, we need to determine one fog node, either
one helper node or the task node, to deal with it. This corresponds to choosing
one arm to play in the MAB.

Recall that the task node generates one task at the beginning of each time slot.
Let $\tau_s\le s + \tau_{\rm max}$ be the time when the feedback of the $s$-th
task is received, where $\tau_{\rm max}$ is the maximum permitted latency.
If $\tau_s>s + \tau_{\rm max}$, the task fails and is discarded. According to
\cite{2011_non-stationary_UCB}, we can estimate $W_t(i)$ and $P_t(i)$
with the UCB policy as
\begin{equation}\label{Eq:W_P_est}
\begin{split}
\bar W_t(\gamma,i):=&\frac{1}{N_t(\gamma,i)} \sum_{s=1}^{t} \gamma^{t-\tau_s} w_s(i) \mathds 1{\{I_s=i, \tau_s\le t\}},\\
\bar P_t(\gamma,i):=&\frac{1}{N_t(\gamma,i)} \sum_{s=1}^{t} \gamma^{t-\tau_s} p_s(i) \mathds 1{\{I_s=i, \tau_s\le t\}},
\end{split}
\end{equation}
where $\gamma\in(0,1)$ represents the discount factor, and
\begin{equation}\label{Eq:N_est}
	N_t(\gamma,i):=  \sum_{s=1}^{t} \gamma^{t-\tau_s}  \mathds 1{\{I_s=i, \tau_s\le t\}}.
\end{equation}
Then the latency $U_t(i)$ can be estimated as
\begin{equation}\label{Eq:Delay_est}
\bar \mu_t(\gamma,i):= L_t  T(i) + { Q_t(i)  \bar W_t(\gamma,i)
 + L_t \bar P_t(\gamma,i) }.
\end{equation}
Note that the latency in (\ref{Eq:Delay_est}) is estimated based on the history
information of $W_s(i)$ and $P_s(i)$ instead of the previous latency values, i.e.
$U_s(i), s<t$. This is due to the fact that the individual latency closely depends
on the queue length $Q_t$ and the task length $L_t$, which may vary significantly
for different types of tasks. Thus it is not trustworthy to estimate $U_t(i)$
with the previous latency values directly. On the other hand, the time needed to
process one bit of a task is typically determined by the node capability, which is
relatively stable and thus suitable to be estimated with the sample mean.

At node-$i$, the total amount of time utilized to process task-$k$ is compared
with the maximal tolerable latency and the time difference is
defined as a reward, i.e. $X_t(i) := \tau_{\rm max} - U_t(i)$.
Clearly, a negative reward indicates a task failure. Based on the estimated
latency in (\ref{Eq:Delay_est}), the estimated reward is given by
\begin{equation}\label{Eq:X_est}
	\bar X_t(\gamma,i) := \tau_{\rm max} -  \bar \mu_t(\gamma,i).
\end{equation}
The parameters $N_{t+1}(\gamma,i)$, $\bar W_{t+1}(\gamma,i)$, and
$\bar P_{t+1}(\gamma,i)$ can be updated iteratively with low complexity.
Particularly, let $\mathcal S_t := \{s|t < \tau_{s} \le t+1\}$
denote the set of indices of tasks completed within the interval $(t,t+1]$
and we can have
\begin{equation}\label{Eq:Update_N}
	N_{t+1}(\gamma,i) = \gamma N_t(\gamma,i) + \sum_{s\in\mathcal S_t} \gamma^{t+1-\tau_s} \mathds 1{\{I_s = i\}},
\end{equation}
\begin{equation}\label{Eq:Update_W}
	\begin{split}
		\bar W_{t+1}(\gamma,i)   =  \frac{1}{N_{t+1}(\gamma,i)}&\Big[\gamma  N_{t}(\gamma,i)\bar W_{t}(\gamma,i) \\
	  + \sum_{s\in \mathcal S_t} &\gamma^{t+1-\tau_s} w_s(i)\mathds 1{\{I_s = i\}}\Big],
	\end{split}
\end{equation}
\begin{equation}\label{Eq:Update_P}
	\begin{split}
		\bar P_{t+1}(\gamma,i)   =  \frac{1}{N_{t+1}(\gamma,i)}&\Big[\gamma  N_{t}(\gamma,i)\bar P_{t}(\gamma,i) \\
	  + \sum_{s\in \mathcal S_t} &\gamma^{t+1-\tau_s} p_s(i)\mathds 1{\{I_s = i\}}\Big].
	\end{split}
\end{equation}

The exploration-exploitation tradeoff is then handled by applying the UCB policies
as in \cite{2011_non-stationary_UCB}. An UCB is constructed as
$\bar X_t(\gamma,i) + c_t(\gamma,i)$. The padding function $c_t(\gamma,i)$
characterizes the exploration bonus, which is defined as
\begin{equation}\label{Eq:c_est}
c_t(\gamma,i) := 2 \tau_{\rm max}
  \sqrt{\frac{\xi \log n_t(\gamma)}{N_t(\gamma,i)}},
\end{equation}
where $\xi$ stands for an exploration constant and
\begin{equation}\label{Eq:n_est}
   n_t(\gamma) := \sum_{i=1}^{K} N_t(\gamma,i).
\end{equation}
The node selected to process task-$t$ is then determined by
\begin{equation}\label{Eq:Sel_Node}
	 I_t = \arg\max_{i\in\mathcal I}  \ \bar X_t(\gamma,i) + c_t(\gamma,i).
\end{equation}
Our proposed strategy, i.e. Task Offloading with Discounted-UCB (TOD),
is summarized in Algorithm \ref{Alg_TOD}.

Although the above proposed task offloading model is essentially a non-stationary
MAB model, there are two main differences compared with the conventional model
as proposed in \cite{2011_non-stationary_UCB}.
First, the feedback was obtained instantaneously with the decision making
in the conventional model. While in our model, as indicated in (\ref{Eq:W_P_est}),
the feedback is not available until the task is finished.
The corresponding latency should not be ignored since it is exactly the
information we need. Note the delayed feedback affects the performance
analyses as discussed in \cite{2013_DelayFeedback_Online}. Second, the best arm
is assumed to change only at the breakpoints in \cite{2011_non-stationary_UCB}.
However, our model allows the best node to vary when processing different tasks.
Therefore, the performance guarantee for the conventional discounted-UCB algorithm
cannot be applied directly to our proposed TOD.

\begin{algorithm}[t]
\caption{TOD (Task Offloading with Discounted-UCB) Algorithm}
\begin{algorithmic}[1]
\State {\bf Initialization:} Set appropriate $\gamma$. Set $t=1$, $\bar X_t(\gamma,i) = \bar W_t(\gamma,i)=\bar P_t(\gamma,i)=0, \forall i\in\mathcal I$.
\State {\bf Repeat}
\State \quad Let $I_t = t$, offload task-$t$ to node-$I_t$; $t=t+1$;
\State {\bf Until}  $t>K$;
\State Update  $\bar W_t(\gamma,i)$, $\bar P_t(\gamma,i)$ and $\bar N_t(\gamma,i)$ as (\ref{Eq:W_P_est}) and (\ref{Eq:N_est});
\State {\bf Repeat}
\State \quad Update  $\bar X_t(\gamma,i)$ and $\bar c_t(\gamma,i)$ as (\ref{Eq:X_est}) and (\ref{Eq:c_est});
\State \quad Determine $I_t$ as (\ref{Eq:Sel_Node}), offload task-$t$ to node-$I_t$;
\State \quad $t=t+1$;
\State \quad Update  $\bar N_t(\gamma,i)$, $\bar W_t(\gamma,i)$, and $\bar P_t(\gamma,i)$ as (\ref{Eq:Update_N})-(\ref{Eq:Update_P});
\State {\bf Until}  $t>T$;
\end{algorithmic}
\label{Alg_TOD}
\end{algorithm}

\subsection{Performance Analysis}

According to (\ref{Eq:Expect_W_P}) and (\ref{Eq:Latency}), the expected latency
$\mathbb E[U_t(i)]$ can be expressed as
\begin{equation}
   \mu_t(i):= \mathbb E[U_t(i)] = L_t T(i) + Q_t(i) \mu^W_t(i) + L_t \mu^P_t(i).
\end{equation}
Given the offloading strategies for the first $(t-1)$ tasks, according to
(\ref{Eq:OptEach}), the best node to handle task-$t$ is given by
$i_t^*:= \arg \min_{i\in\mathcal I}  \mu_t(i)$.
Additionally, we use
$$\tilde N_T(i):= \sum_{t=1}^T \mathds 1{\{I_t =i \ne i_t^*\}}$$
to denote the number of tasks offloaded to node-$i$ while it is not the best
node during the first $T$ time slots. From AS-3, we know the expectations of
system parameters could change abruptly at each breakpoint.
We use $\Upsilon_T$ to denote the number of breakpoints before time $T$.
The following proposition provides an upper bound for $\mathbb E(\tilde N_T(i))$.

\begin{pro}\label{Pro_BoundTimes}
Assume $\xi>1/2$ and $\gamma\in(0,1)$ satisfies
\begin{equation*}
{\gamma^{\tau_{\rm max}} (1-\gamma^{1/(1-\gamma)})} /{(1-\gamma)} >e.
\end{equation*}
For each node $i\in{\mathcal I}$, we have the following upper bound for
$\mathbb E(\tilde N_T(i))$:
\begin{equation}
\label{Eq:BoundTimes}
\mathbb E \left[\tilde N_T(i)\right] \le 1 + T(1-\gamma) B(\gamma) +
\Upsilon_T C(\gamma) + \frac{2}{1-\gamma},
\end{equation}
where
\begin{eqnarray}
C(\gamma)&:=& \log_\gamma((1-\gamma)\xi\log n_K(\gamma)) + \tau_{\rm max},
\nonumber\\
B(\gamma) &:= & \Big( \frac{-16 \tau_{\rm max}^2 \xi \log {[\gamma^{\tau_{\rm max}}  (1-\gamma)]} } {(\Delta\mu_T(i))^{2}} + \tau_{\rm max}\Big )\nonumber \\
&& \cdot \frac{\lceil T(1-\gamma)  \rceil}{T(1-\gamma)} \gamma^{-\frac{1}{1-\gamma}} +  \frac{2}{\gamma^{\tau_{\rm max}} } {\log \frac{\gamma^{\tau_{\rm max}} } {1-\gamma}},
\nonumber \\
\Delta \mu_T(i) &:=& \min_{t\in\{1,\cdots,T\}, i^*_t\ne i}  \ \mu_t(i)-\mu_t(i^*_t).
\nonumber
\end{eqnarray}
\end{pro}
Detailed proof for the above proposition can be found in arXiv and is omitted here
due to the space limitation\footnote{arXiv:1804.08416,
https://arxiv.org/abs/1804.08416.}.
Clearly, the upper bound depends on the number of total tasks, the number of
breakpoints $\Upsilon_T$, and the choice of discount factor $\gamma$.
From (\ref{Eq:BoundTimes}), we see the term $T(1-\gamma) B(\gamma)$ decreases
as the feasible $\gamma$ increases. On the other hand, the last two terms,
i.e. $\Upsilon_T C(\gamma)$ and $2/(1-\gamma)$, are increasing when the feasible
$\gamma$ is increasing. This is consistent with our intuition that a higher
discount factor $\gamma$ contributes to a better estimation in the stationary case,
while it results in slow reaction to abrupt changes of environments. Therefore,
there is a tradeoff between different terms in (\ref{Eq:BoundTimes}).
To strike a balance between the stable and the abruptly-changing environments,
similar to \cite{2011_non-stationary_UCB}, we choose $\gamma$ as
\begin{equation}\label{Eq:Sel_Gamma}
	\gamma = 1-(4\tau_{\rm max})^{-1}\sqrt{\Upsilon_T/T}.
\end{equation}
Accordingly, we can establish the following proposition.

\begin{pro}\label{Eq:BoundTimesOrder}
When $\Upsilon_T=\mathcal O(T^\beta), \beta\in[0,1)$, and $T\rightarrow\infty$,
the value of $\mathbb E(\tilde N_T(i))$ is in the order of
\begin{equation*}
\mathbb E \left[\tilde N_T(i)\right] = \mathcal{O}\left( \sqrt{T\Upsilon_T}
\log T \right).
\end{equation*}
\end{pro}
\begin{proof}
	Let $\gamma = 1-(4\tau_{\rm max})^{-1}\sqrt{\Upsilon_T/T}$, then the three terms in (\ref{Eq:BoundTimes}), i.e. $T(1-\gamma)B(\gamma)$, $\Upsilon_T C(\gamma)$, and $2/(1-\gamma)$, are in the order of $\mathcal O(\sqrt{\Upsilon_T T}\log T)$, 
$\mathcal O(\sqrt{\Upsilon_T T}\log T)$, and $\mathcal O(\sqrt{T/\Upsilon_T})$, respectively. Thus $\mathbb E [\tilde N_T(i)]$ is in the order of
$ \mathcal{O}( \sqrt{T\Upsilon_T} \log T )$.
\end{proof}

To show the optimality of our proposed Algorithm \ref{Alg_TOD}, we define the
pseudo-regret in offloading the first $T$ tasks as \cite{2012_Bubeck_Bandit}
\begin{equation}
\label{Eq:Regret}
\zeta_T :=\frac{1}{T} \mathbb E
\left[ \sum_{t=1}^T \left( U_t(I_t) - \mathbb E [U_t(i_t^*)] \right) \right].
\end{equation}
We have the following result regarding the pseudo-regret $\zeta_T$.

\begin{pro}
When $\Upsilon_T=\mathcal O(T^\beta), \beta\in[0,1)$, the proposed approach
in Algorithm \ref{Alg_TOD} is asymptotically optimal in the sense that
 $\lim_{T\rightarrow\infty} \zeta_T \overset{\rm a.s.}{\longrightarrow} 0$.
\end{pro}
\begin{proof}
Note $U_t(I_t) - \mathbb E [U_t(i_t^*)]  \le
\tau_{\rm max} \mathds 1{\{I_t\ne i_t^*\}}$.
We have
\begin{equation}
\begin{split}
	\zeta_T &\le \frac{1}{T} \mathbb E \left[ \sum_{t=1}^T \tau_{\rm max} \mathds 1{\{I_t\ne i_t^*\}}   \right] \\
	&\le \frac{\tau_{\rm max}}{T}  \sum_{i\in\mathcal I} \mathbb E \left[ \tilde N_T(i) \right].
\end{split}
\end{equation}
According to Proposition \ref{Pro_BoundTimes} and Proposition
\ref{Eq:BoundTimesOrder}, we obtain
\begin{equation}
\zeta_T = \mathcal{O}\left( \sqrt{\Upsilon_T/T} \log T \right)
= \mathcal{O}\left(  T^{\frac{\beta-1}{2}} \log T \right).
\end{equation}
Then for any $\varepsilon>0$, there exists a finite integer $N_\varepsilon$,
such that
\begin{equation}
\mathbb P(|\zeta_T| \ge \varepsilon) = 0, \ \forall T\ge  N_\varepsilon.
\end{equation}
Therefore,
\begin{equation}
\begin{split}
\sum_{T=1}^{\infty} \mathbb P(|\zeta_T| \ge \varepsilon)
\le N_\varepsilon < \infty.
\end{split}
\end{equation}
The above equation indicates $\zeta_T\overset{\rm a.s.}{\longrightarrow} 0$.
\end{proof}

\section{Numerical Results} \label{Simulation}

In this section, we evaluate the performance of our proposed offloading algorithm by testing $10,000$ rounds of task offloading.
One task is generated in each round.
Some common system parameters are set as follows.\\
\noindent$\bullet$ The network consists of $1$ task node and $9$ helper nodes;\\
\noindent$\bullet$ Each time slot is $20$ ms. Data size follows ${\rm Unif}(1,15)$ KB;\\
\noindent$\bullet$ Maximal latency is $\tau_{\rm max}=20$ slots,  $\xi = 0.6$;\\
\noindent$\bullet$ The delay of processing one bit of the task is simulated following $P_t(i) = \sigma^{\rm cplx}_t/\sigma^{\rm CPU}_i$, where $\sigma^{\rm cplx}_t$ characterizes the complexity of task-$t$, and $\sigma^{\rm CPU}_i$ reflects the CPU capability of node-$i$.
Both variables follow ${\rm Unif}(1,10)$;\\
\noindent$\bullet$ The CPU capability of node-$i$ is changed as $\sigma^{\rm CPU}_i = \sigma^{\rm CPU}_i \times 16$ or $\sigma^{\rm CPU}_i = \sigma^{\rm CPU}_i /16$ at each breakpoint.

We compare the performance of TOD with two other schemes, i.e. {\it Greedy} and {\it Round-Robin}.
In the greedy scheme, we assume full information of every realization and offload the task to the node achieving minimal latency in each time slot.
Note that the greedy scheme is not causal and cannot be applied in practice.
In the round-robin scheme, each task is offloaded to the fog nodes in a cyclic way with equal chances.

\begin{figure}
\centering
 \subfigure[$\Upsilon_T = 150$.]
 {\epsfig{file=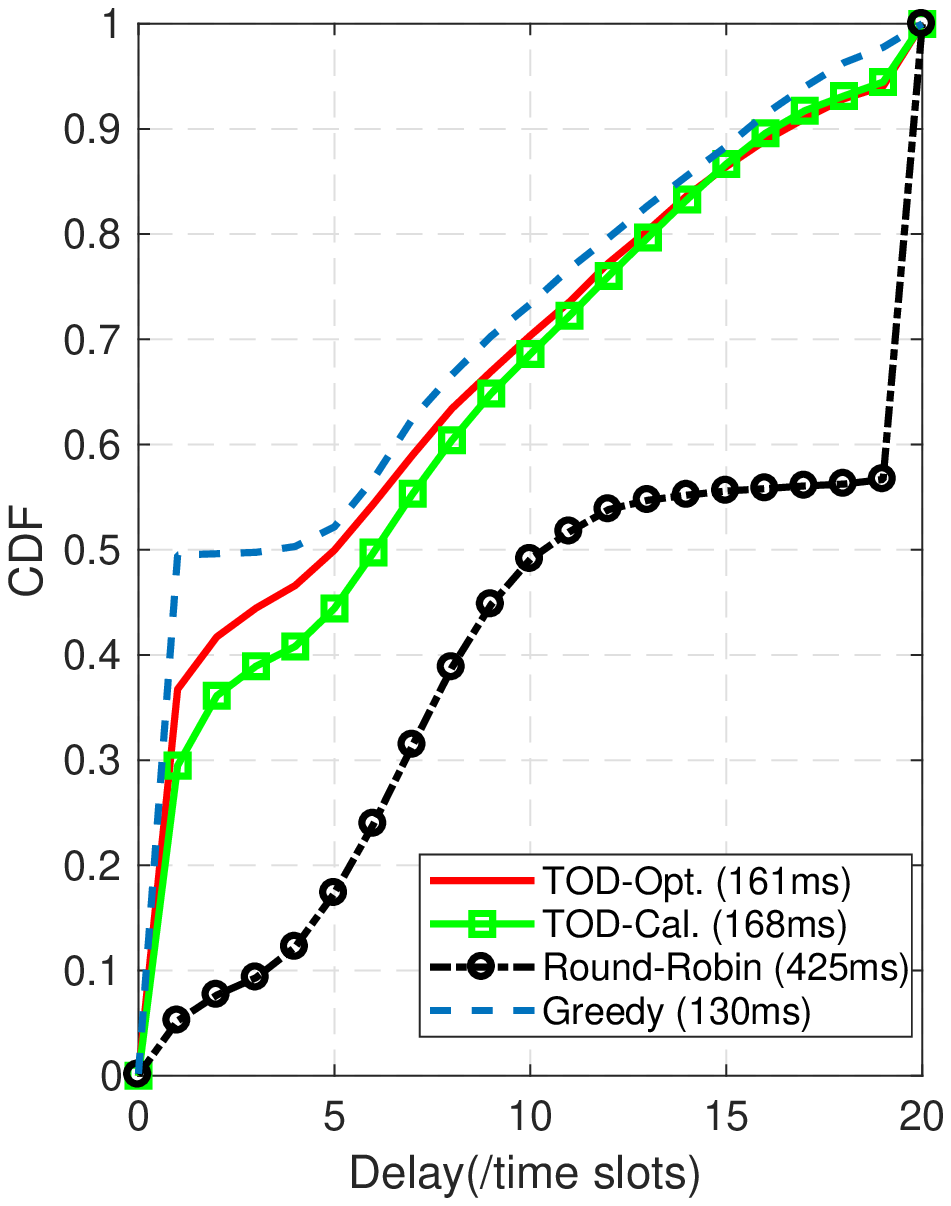,width=0.22\textwidth}
 \label{Fig:CDF_150}}
  \subfigure[$\Upsilon_T = 10$.]
 {\epsfig{file=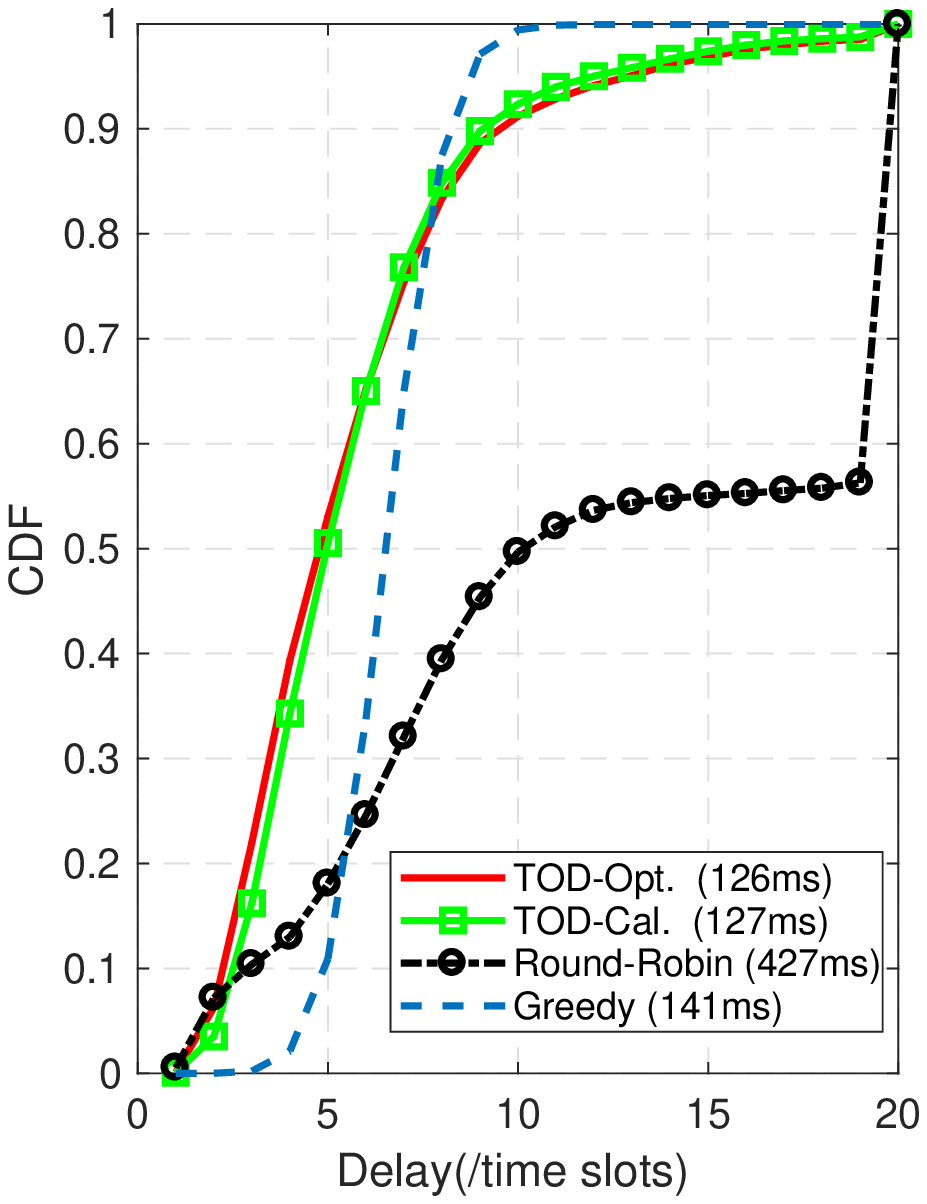,width=0.22\textwidth}
 \label{Fig:CDF_10}}
\caption{CDFs of the latency of processing $10,000$ tasks. TOD-Opt.: solution with TOD using $\gamma_{\rm opt} = 0.9993$ (left), $\gamma_{\rm opt} = 0.9995$ (right); TOD-Cal.: solution with TOD using $\gamma_{\rm cal} = 0.9985$ (left), $\gamma_{\rm cal} = 0.9996$ (right).
}
\label{Fig:CDF}
 \vspace{-5pt}
\end{figure}

\begin{figure}[t]
\centering
\includegraphics[width = 0.45\textwidth]{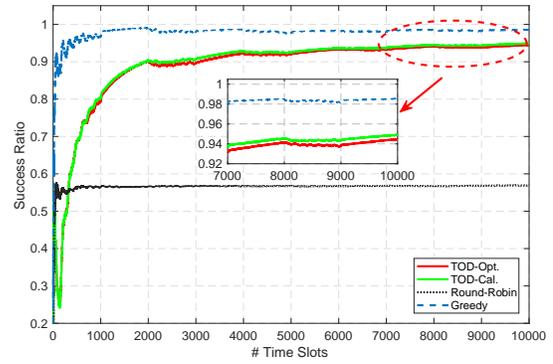}
\vspace{-5pt}
\caption{Cumulative success ratio versus time. The number of breakpoints is set as $\Upsilon_T = 150$.
}
\label{Fig:SuccessCntPer}
 \vspace{-10pt}
\end{figure}

\begin{figure}[t]
\centering
\includegraphics[width = 0.45\textwidth]{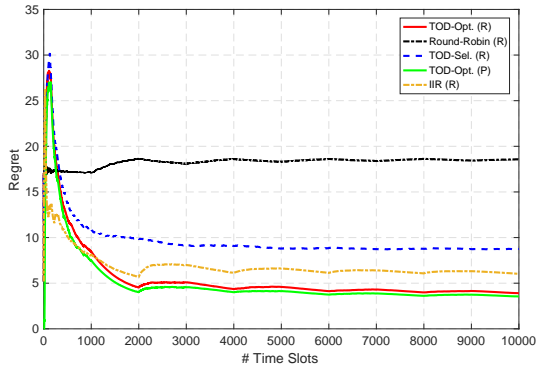}
\vspace{-10pt}
\caption{Average regret versus time.
The regret is calculated as $\hat \zeta_T^{\rm} :=\frac{1}{T} \sum_{t=1}^T ( U_t(I_t) - U_t(i_t) )$.
	(R): Regret by taking $i_t= \arg \max_{i\in\mathcal I} U_t(i)$; (P): Pseudo-regret by taking $i_t = \arg \max_{i\in\mathcal I} \mu_t(i)$.
	TOD-Sel.: solution with TOD using $\gamma_{\rm sel} = 0.98$;
	IIR: solution with separated exploration and exploitation using the same discount factor as TOD-Opt.
	The number of breakpoints is set as $\Upsilon_T = 150$.}
\label{Fig:Regret}
 \vspace{-8pt}
\end{figure}

In Fig.~\ref{Fig:CDF}, with different number of breakpoints, i.e $\Upsilon_T$, we demonstrate the effectiveness and robustness of TOD by showing cumulative distribution functions (CDFs) of the latency of processing $10,000$ tasks with different schemes.
{\it TOD-Opt} and {\it TOD-Cal} in Fig.~\ref{Fig:CDF} represent two different criteria to choose $\gamma$ in TOD.
The discount factor $\gamma$ in the former criterion is searched over $(0,1)$ to achieve the minimal average latency, while the other one is calculated following (\ref{Eq:Sel_Gamma}).
Both the left and the right parts in Fig.~\ref{Fig:CDF} show the proposed TOD algorithm performs much better than the round-robin scheme, and performs close to the greedy method, which achieves the minimal realization of latency in each round.
Additionally, we can learn from Fig.~\ref{Fig:CDF} that the $\gamma$ calculated following (\ref{Eq:Sel_Gamma}) performs as well as the optimal one.
In Fig.~\ref{Fig:CDF_150}, there exists one breakpoint every $67$ tasks on average as $\Upsilon_T=150$, which indicates TOD is able to learn the system under frequent changes of parameter distribution.
It is also worth noting that, in Fig.~\ref{Fig:CDF_10}, TOD achieves even less average latency than the greedy scheme in the case of limited abrupt changes, i.e. $\Upsilon_T=10$.
This phenomenon reveals that the decision with minimal latency in each time slot may not be the global optimum of (\ref{Eq:GlobalOpt}).
It also corroborates our previous analysis that every choice will affect the future state of the node, and further affects the following offloading decisions.

Fig.~\ref{Fig:SuccessCntPer} presents the ratio of the number of successfully processed tasks to the number of tasks.
A task is successful if the latency is less than $\tau_{\rm max}$.
Although the success ratios of TOD are lower than the greedy scheme due to the exploration of nodes, they show the tendencies to approaching the greedy scheme with time going on.

Fig.~\ref{Fig:Regret} depicts the regrets from different schemes.
Note the regret is based on the optimal realization (greedy method), and the pseudo-regret is based on the optimal expectation.
In {\it IIR}, we separate the exploration and the exploitation to two phases.
In the exploration phase, the round-robin method is adopted.
In the exploitation phase, we focus on maximizing the estimated reward defined in (\ref{Eq:X_est}), which is actually an estimate based on the infinite impulse response (IIR) filter.
The ratio of two phases is searched to achieve the minimal regret.
It can be observed that, either in the sense of the regret or in the sense of
the pseudo-regret, the proposed TOD algorithm achieves much lower regrets than the
round-robin scheme and the {\it IIR} scheme. 
This phenomenon shows that our proposed method performs well in dealing with the exploration-exploitation tradeoff.
Besides, as the discount factor is set to $\gamma_{\rm sel}$, the TOD performance deteriorates a lot.
This further indicates the importance of the exploration bonus.

\section{Conclusion} \label{Conclusion}

In this paper, an efficient online task offloading strategy and the corresponding
performance guarantee in a fog-enabled network have been studied. Considering that
the expectations of processing speeds change abruptly at unknown time instants,
and the system information is available only after finishing the corresponding tasks, we
have formulated it as a stochastic programming with delayed bandit feedbacks.
To solve this problem, we have provided TOD, an efficient online task offloading
algorithm based on the UCB policy. Given a particular number of breakpoints
$\Upsilon_T$, we have proven that the bound on the number of tasks offloaded to a
particular non-optimal node is in the order of $\mathcal O(\sqrt{\Upsilon_T T}\log T)$.
Besides, we have also proven that the pseudo-regret goes to zero almost surely when
the number of tasks goes to infinity. Simulations have demonstrated that the proposed
TOD algorithm is capable of learning and picking the right node to offload tasks
under non-stationary circumstances.

\section{Appendix}

\begin{proof}
According to the definition of $\tilde N_T(i)$, it can be decomposed as
\begin{equation}
\begin{split}
		\tilde N_T(i)  \le & 1 + \sum_{t= K + 1}^T \mathds 1{\{I_t = i \ne i^*_t, N_t(\gamma,i) < A(\gamma,i)\}} \\
		& + \sum_{t= K + 1}^T \mathds 1{\{I_t = i \ne i^*_t, N_t(\gamma,i) \ge A(\gamma,i)\}},
\end{split}
\end{equation}
where $A(\gamma,i)$ is a particular function with respect to $\gamma$.
The number of missing feedbacks when task-$t$ is offloaded can be defined as
\begin{equation}
	G_t(i):= \sum_{s=1}^{t-1}  \mathds 1{\{I_s=i\}} - N_t(1,i).
\end{equation}
Clearly, the number of missing feedbacks is no larger than $\tau_{\rm max}$, i.e. $G_t(i)\le \tau_{\rm max}, \forall t,i$.
According to Lemma 1 in \cite{2011_non-stationary_UCB}, for any $i\in\mathcal I, \tau>0, m>0$, the following inequality is derived:
\begin{equation}
	\sum_{t=K+1}^T \mathds 1{\{I_t = i, \sum_{s=t-\tau}^{t-1} \mathds 1{\{I_s=i\}}<m\}} \le \lceil {T}/{\tau} \rceil m.
\end{equation}
Due to the fact that
\begin{equation}
	\begin{split}
		& \left\{t|\sum_{s=t-\tau}^{t-1} \mathds 1{\{I_s=i\}}<m \right \} \\
		 =& \left\{t|\sum_{s=t-\tau}^{t-1} \mathds \gamma^\tau 1_{\{I_s=i\}}<\gamma^\tau m \right \} \\
		 \supseteq &\left\{t|   \sum_{s=t-\tau}^{t-1} \gamma^{t-\tau_s}  \mathds 1{\{I_s=i, \tau_s\le t\}} + G_t(i)<\gamma^\tau m \right\} \\
		 \supseteq &\left\{t|N_t(\gamma,i)+\tau_{\rm max}<\gamma^\tau m \right\},
	\end{split}
\end{equation}
we have
\begin{equation}
	\sum_{t=K+1}^T \mathds 1{\{I_t = i, N_t(\gamma,i)+\tau_{\rm max}<\gamma^\tau m \}} \le \lceil {T}/{\tau} \rceil m.
\end{equation}
Let $m=\gamma^{-\tau}( A(\gamma,i) + \tau_{\rm max})$, we have
\begin{equation}
\begin{split}
	&\sum_{t=K+1}^T \mathds 1{\{I_t = i\ne i^*_t, N_t(\gamma,i)<A(\gamma,i) \}} \\
	 \le& \lceil {T}/{\tau} \rceil \gamma^{-\tau}( A(\gamma,i) + \tau_{\rm max}).
\end{split}
\end{equation}
Let $\Upsilon_T$ denote the number of breakpoints before time $T$, and $\mathcal T(\gamma)$ denote the set of ``well offloaded'' tasks.
Mathematically, these tasks are defined as follows.
\begin{equation}
\begin{split}
	\mathcal T(\gamma):=&\{t|   t\in \{K+1,\cdots,T\}; \\
	& \mu_s(j)=\mu_t(j), \forall s\in (t-C(\gamma),t), \forall j \in \mathcal I \},
\end{split}
\end{equation}
where $C(\gamma)$ indicates the number of tasks, of which the delay is poorly estimated.
Because of this, the D-UCB policy may not offload tasks to the optimal node, which leads to the following bound:
\begin{equation}
\label{Eq:BreakPoints_result}
\begin{split}
&\sum_{t= K + 1}^T \mathds 1{\{I_t = i \ne i^*_t, N_t(\gamma,i) \ge A(\gamma,i)\}} \\
\le & \Upsilon_T C(\gamma) +  \sum_{t\in\mathcal T(\gamma)} \mathds 1{\{I_t = i \ne i^*_t, N_t(\gamma,i) \ge A(\gamma,i)\}}.
\end{split}
\end{equation}
Next, we need to upper-bound the last term in (\ref{Eq:BreakPoints_result}).
There are three facts:\\
i) The event $\{I_t = i \ne i^*_t\}$ occurs if and only if the event $\mathcal E_t(\gamma,i) = \{\bar \mu_t(\gamma,i_t^*)-c_t(\gamma,i_t^*) \ge \bar \mu_t(\gamma,i)-c_t(\gamma,i)\}$ occurs; \\
ii) $\mathcal E_t(\gamma,i) \subseteq \{\bar \mu_t(\gamma,i_t^*)-c_t(\gamma,i_t^*) \ge \mu_t(i_t^*) \} \cup \{\bar \mu_t(\gamma,i)-c_t(\gamma,i) < \mu_t(i_t^*) \}$; \\
iii) $\{\bar \mu_t(\gamma,i)-c_t(\gamma,i) < \mu_t(i_t^*) \}\subseteq \{\bar \mu_t(\gamma,i) + c_t(\gamma,i) <  \mu_t(i) \} \cup\{\mu_t(i)-\mu_t(i_t^*) < 2c_t(\gamma,i)  \}$.\\
Based on these facts, the following inequality is obtained:
\begin{equation}
	\begin{split}
		  & {\{I_t = i \ne i^*_t, N_t(\gamma,i) \ge A(\gamma,i)\}} \\
\subseteq & {\{\bar \mu_t(\gamma,i_t^*)-c_t(\gamma,i_t^*) \ge \mu_t(i_t^*) \}} \\
		  & \cup {\{\bar \mu_t(\gamma,i) + c_t(\gamma,i) <  \mu_t(i) \}} \\
		  & \cup {\{\mu_t(i)-\mu_t(i_t^*) < 2c_t(\gamma,i),N_t(\gamma,i) \ge A(\gamma,i)  \}}.
	\end{split}
\end{equation}
Namely, when node-$i$ is tested enough times by the task node, the event $\{I_t = i \ne i^*_t\}$ only occurs under three circumstances: i) the delay of the optimal node is substantially overestimated; ii) the delay of node-$i$ is substantially underestimated; iii) both delay expectations, i.e. $\mu_t(i^*_t)$ and $\mu_t(i)$, are close enough.

However, if $A(\gamma,i)$ is chosen appropriately, the event $\{\mu_t(i)-\mu_t(i_t^*) < 2c_t(\gamma,i),N_t(\gamma,i) \ge A(\gamma,i)  \}$ never occurs.
Denote the minimal difference between the expected delay of node-$i$ and the expected delay of the best node-$i^*$ by $\Delta \mu_T(i)$, i.e.
	\begin{equation}
		\Delta \mu_T(i) := \min_{t\in\{1,\cdots,T\}, i^*_t\ne i}  \ \mu_t(i)-\mu_t(i^*_t).
	\end{equation}
Let $A(\gamma,i):= 16 \tau_{\rm max}^2 \xi \log n_{t^*}(\gamma) (\Delta\mu_T(i))^{-2}$, where $t^* = \arg \max_{t\in\{1,\cdots,T\}} n_t(\gamma)$.
Recalling $N_t(\gamma,i) \ge A(\gamma,i)$, we have
\begin{equation}
\label{Eq:Contradict_1}
	\frac{\Delta\mu_T(i)}{2} = 2 \tau_{\rm max} \sqrt{\frac{ \xi \log n_{t^*}(\gamma) }{A(\gamma,i)}}
	\ge   c_t(\gamma,i).
\end{equation}
However, from the definition of $\Delta \mu_T(i)$ we obtain:
\begin{equation}
	\frac{\Delta \mu_T(i)}{2} \le \frac{\mu_t(i) - \mu_t(i^*_t)}{2} <  c_t(\gamma,i),
\end{equation}
which is contradict with (\ref{Eq:Contradict_1}).
Thus the events $\{\mu_t(i)-\mu_t(i_t^*) < 2c_t(\gamma,i)\}$ and $\{N_t(\gamma,i) \ge A(\gamma,i)\}$ never occur simultaneously, which indicates that we only need to upper-bound the probability of events $\{\bar \mu_t(\gamma,i_t^*)-c_t(\gamma,i_t^*) \ge \mu_t(i_t^*) \}$ and $\cup {\{\bar \mu_t(\gamma,i) + c_t(\gamma,i) <  \mu_t(i) \}}$.
Define $M_t(\gamma,i)$ as
\begin{equation}
\begin{split}
M_t(\gamma,i):=  \sum_{s=1}^{t} \gamma^{t-\tau_s} m_t(s,i)\mathds 1{\{I_s=i, \tau_s\le t\}},
\end{split}
\end{equation}
where $m_t(s,i)  = \big( L_t T(i) + Q_t(i) \mu^W_s(i)  + L_t\mu^P_s(i)\big)$,
then
\begin{equation}
	\begin{split}
		&|M_t(\gamma,i) - \mu_t(i)N_t(\gamma,i)| \\
		=& \left| \sum_{s=1}^{t-C(\gamma)}\gamma^{t-\tau_s} \left(m_t(s,i)-\mu_t(i) \right)\mathds 1{\{I_s=i, \tau_s\le t\}}  \right| \\
		 \le &\sum_{s=1}^{t-C(\gamma)}\gamma^{t-\tau_s} \left|m_t(s,i)-\mu_t(i) \right|\mathds 1{\{I_s=i, \tau_s\le t\}} \\
		 \le & \tau_{\rm max}\sum_{s=1}^{t-C(\gamma)}\gamma^{t-\tau_s} \mathds 1{\{I_s=i, \tau_s\le t\}} \\
		 \le & \tau_{\rm max}\sum_{s=1}^{t-C(\gamma)}\gamma^{t-s-\tau_{\rm max}} \mathds 1{\{I_s=i\}} \\
		 \le &\tau_{\rm max} \gamma^{C(\gamma)-\tau_{\rm max}} (1-\gamma)^{-1}.
	\end{split}
\end{equation}
Combining with the following two facts:
\begin{equation}
	\left| \frac{M_t(\gamma,i)}{N_t(\gamma,i)} - \mu_t(i) \right| \le \tau_{\rm max},  \  \min(1,x) \le \sqrt{x}, \forall x\ge 0,
\end{equation}
we obtain
\begin{equation}
\label{Eq:Bound_bias}
	\left| \frac{M_t(\gamma,i)}{N_t(\gamma,i)} - \mu_t(i) \right| \le \tau_{\rm max} \sqrt{ \frac{\gamma^{C(\gamma)-\tau_{\rm max}} } {(1-\gamma)N_t(\gamma,i)}}.
\end{equation}
Let
\begin{equation}
	C(\gamma):= \log_\gamma((1-\gamma)\xi\log n_K(\gamma)) + \tau_{\rm max},
\end{equation}
the inequality in (\ref{Eq:Bound_bias}) turns to be
\begin{equation}
	\left| \frac{M_t(\gamma,i)}{N_t(\gamma,i)} - \mu_t(i) \right| \le \frac{1}{2} c_t(\gamma,i).
\end{equation}
Defining $Y_t(\gamma,i):=\bar \mu_t(\gamma,i)$, the following inequality can be deduced:
\begin{equation}
	\begin{split}
	& \mathbb P\left( \mu_t(i)- \bar \mu_t(\gamma,i) >c_t(\gamma,i) \right) \\
	\le & \mathbb P\left( \mu_t(i)-\bar \mu_t(\gamma,i) >  \frac{1}{2}c_t(\gamma,i) +  \left| \frac{M_t(\gamma,i)}{N_t(\gamma,i)} - \mu_t(j) \right| \right)  \\
	 \le & \mathbb P\left( \frac{M_t(\gamma,i)}{N_t(\gamma,i)}-\bar \mu_t(\gamma,i)  > \tau_{\rm max} \sqrt{\frac{ \xi \log n_t(\gamma) }{N_t(\gamma,i)}}   \right)  \\
	= & \mathbb P\left( \frac{M_t(\gamma,i) - Y_t(\gamma,i)}{\sqrt{N_t(\gamma^2,i)}} > \tau_{\rm max} \sqrt{\frac{ \xi N_t(\gamma,i) \log n_t(\gamma) }{N_t(\gamma^2,i)}}   \right)  \\
	\le & \mathbb P\left( \frac{M_t(\gamma,i) - Y_t(\gamma,i)}{\sqrt{N_t(\gamma^2,i)}} > \tau_{\rm max} \sqrt{{ \xi \log n_t(\gamma) }}   \right)  \\
	\overset{(a)}{\le} & \left\lceil\frac{\log n_t(\gamma)}{\log(1+\eta)} \right\rceil \exp \left( -{2 \xi \log n_t(\gamma)} \left(1-\frac{\eta^2}{16}\right)  \right),
	\end{split}
\end{equation}
where $(a)$ holds due to Theorem 4 in \cite{2011_non-stationary_UCB}.
Let $\eta:= 4\sqrt{1-1/(2\xi)}, \xi >1/2$, we further obtain:
\begin{equation}
 	\mathbb P\left( \mu_t(i)-\bar \mu_t(\gamma,i) > c_t(\gamma,i) \right) \le \left\lceil\frac{\log n_t(\gamma)}{\log(1+\eta)} \right\rceil  n_t(\gamma)^{-1}.
 \end{equation}
 Till now, the expectation of $\tilde N_T(i)$ can be upper-bounded as
\begin{equation}
\begin{split}
\mathbb E \left[\tilde N_T(i)\right] \le & 1 + \lceil {T}/{\tau} \rceil \gamma^{-\tau}( A(\gamma,i) + \tau_{\rm max} )
+ \Upsilon_T C(\gamma) \\
& + 2 \sum_{t\in\mathcal T(\gamma)} \left\lceil\frac{\log n_t(\gamma)}{\log(1+\eta)} \right\rceil n_t(\gamma)^{-1}.
\end{split}
\end{equation}
Assuming
\begin{equation}
	\sum_{s=1}^{\tau} \gamma^{\tau-s+\tau_{\rm max}} =\frac{\gamma^{\tau_{\rm max}} (1-\gamma^{\tau})} {1-\gamma} >e, \ \ \tau = (1-\gamma)^{-1},
\end{equation}
we have
\begin{equation}
	n_t{(\gamma)} \ge \frac{\gamma^{\tau_{\rm max}} (1-\gamma^{\tau})} {1-\gamma}=\tilde n(\gamma), \forall t\ge \tau ,
\end{equation}
and
\begin{equation}
		\left\lceil\frac{\log n_t(\gamma)}{\log(1+\eta)} \right\rceil n_t(\gamma)^{-1} \le \left\lceil\frac{\log \tilde n(\gamma)}{\log(1+\eta)} \right\rceil \tilde n(\gamma)^{-1}, \forall t\ge \tau.
\end{equation}
Then the following inequality holds:
\begin{equation}
\begin{split}
&\sum_{t\in\mathcal T(\gamma)} \left\lceil\frac{\log n_t(\gamma)}{\log(1+\eta)} \right\rceil n_t(\gamma)^{-1} \\
\le &\tau - K + \sum_{t=\tau}^T \left\lceil\frac{\log \tilde n(\gamma)}{\log(1+\eta)} \right\rceil \tilde n(\gamma)^{-1}\\
\le &\tau - K +  \left\lceil\frac{\log \frac{\gamma^{\tau_{\rm max}} (1-\gamma^{\tau})} {1-\gamma}}{\log(1+\eta)} \right\rceil \frac{T(1-\gamma)}{\gamma^{\tau_{\rm max}} (1-\gamma^{\tau})}. \\
\end{split}
\end{equation}
Therefore, we can upper-bound the expectation of $\tilde N_T(i)$ as
\begin{equation}
\begin{split}
\mathbb E \left[\tilde N_T(i)\right]& \le  1 + \lceil {T}/{\tau}  \rceil \gamma^{-\tau}( A(\gamma,i) + \tau_{\rm max})
+ \Upsilon_T C(\gamma) \\
+ & 2 \left(\tau - K +  \left\lceil\frac{\log \frac{\gamma^{\tau_{\rm max}} (1-\gamma^{\tau})} {1-\gamma}}{\log(1+\eta)} \right\rceil \frac{T(1-\gamma)}{\gamma^{\tau_{\rm max}} (1-\gamma^{\tau})}\right)\\
&\le  1 + \lceil {T(1-\gamma)}  \rceil \gamma^{-\frac{1}{1-\gamma}}( A(\gamma,i) + \tau_{\rm max})
+ \Upsilon_T C(\gamma) \\
+ & 2 \left(\frac{1}{1-\gamma} + \frac{T(1-\gamma)}{\gamma^{\tau_{\rm max}} } {\log \frac{\gamma^{\tau_{\rm max}} } {1-\gamma}}  \right)\\
&\overset{(b)}{\le}  1 + T(1-\gamma) B(\gamma) + \Upsilon_T C(\gamma) + \frac{2}{1-\gamma},
\end{split}
\end{equation}
where $B(\gamma)$ is defined as:
\begin{equation}
\begin{split}
B(\gamma) := &  \left( \frac{-16 \tau_{\rm max}^2 \xi \log {[\gamma^{\tau_{\rm max}}  (1-\gamma)]} } {(\Delta\mu_T(i))^{2}}+ \tau_{\rm max}\right ) \\
 & \cdot \frac{\lceil T(1-\gamma)  \rceil}{T(1-\gamma)} \gamma^{-\frac{1}{1-\gamma}}  +  \frac{2}{\gamma^{\tau_{\rm max}} } {\log \frac{\gamma^{\tau_{\rm max}} } {1-\gamma}},
\end{split}
\end{equation}
and $(b)$ holds since
\begin{equation}
	A(\gamma,i) \le \frac{-16 \tau_{\rm max}^2 \xi \log {[\gamma^{\tau_{\rm max}}  (1-\gamma)]} } {(\Delta\mu_T(i))^{2}}.
\end{equation}

\end{proof}

\end{document}